\documentclass[conference,a4paper]{IEEEtran}
\normalsize
\usepackage[utf8]{inputenc} 
\usepackage[T1]{fontenc}
\usepackage{amsmath}
\usepackage{mathtools}
\usepackage{amsfonts}
\usepackage{amssymb}
\usepackage{amsthm}
\usepackage{subfigure}
\usepackage{cite}
\usepackage{calc}
\usepackage{color}
\usepackage{epsfig}
\usepackage{setspace}
\usepackage{pstricks}
\usepackage{cancel}
\usepackage{multirow}
\usepackage{mathrsfs}
\usepackage[linesnumbered,ruled]{algorithm2e}

\newtheorem{defn}{Definition}
\newtheorem{thm}{{\cal T}heorem}

\newtheorem{constr}{Construction}
\newtheorem{remark}{Remark}

\newtheorem{example}{Example}
\newcommand{\ndiv}{\hspace{-4pt}\not|\hspace{2pt}}

\begin{document}
	\title{An Embedded Index Code Construction Using Sub-packetization}
\author{
	\IEEEauthorblockN{Shanuja Sasi$^{1}$, Vaneet Aggarwal$^{2}$, and B. Sundar Rajan$^{1}$}
	\IEEEauthorblockA{ $^{1}$ Indian Institute of Science, Bengaluru $^{2}$ Purdue University, Indiana\\
		E-mail: shanuja@iisc.ac.in, vaneet@purdue.edu, bsrajan@iisc.ac.in}
}
\maketitle 

\begin{abstract}
	A variant of the index coding problem (ICP), the embedded index coding problem (EICP) was introduced in [A. Porter and M. Wootters, ``Embedded Index Coding,'' ITW, Sweden, 2019] which was motivated by its application in distributed computing where every user can act as sender for other users and an algorithm for code construction was reported. The construction depends on the computation of min-rank of a matrix, which is computationally intensive. In [A.A. Mahesh, N. S. Karat and B. S. Rajan, ``Min-rank of Embedded Index Coding Problems,'' ISIT, 2020], the authors have provided an explicit code construction for a class of EICP - \textit{Consecutive and Symmetric Embedded Index Coding Problem (CS-EICP)}. We introduce the idea of sub-packetization of the messages in index coding problems to provide a novel code construction for CS-EICP in contrast to the scalar linear solutions provided in the prior works. For CS-EICP, the normalized rate, which is defined as the number of bits transmitted by all the users together normalized by the total number of bits of all the messages, for our construction is lesser than the normalized rate achieved by Mahesh \textit{et al.}, for scalar linear codes.
%	 For CS-EICP, if the demand set of each user is neighboring, consecutive and symmetric, then a new lower bound on normalized rate, which is defined as the number of bits transmitted by all the users together normalized by the total bits of all the messages, is also derived for any scheme using sub-packetization which is lesser than the lower bound proposed by Mahesh \textit{et al.}, for scalar linear code. 
%	We provide an explicit code construction for the CS-EICP by appropriately invoking sub-packetization of the messages. 
%	Also, we identify a case where the normalized rate achieved in our scheme meets the lower bound.
%An extension to EICP- task based EICP, introduced by Porter \textit{ et al.}, is also studied, where each user can use transmissions from only one other user. We provide code construction for task based CS-EICP as well when each user demands only one message. For such task based CS-EICP, the rate achieved differs from the rate achieved for CS-EICP by less than one, for half of the cases. 
%If the demand set of each user is neighboring, consecutive and symmetric for those cases, the rate achieved for task based CS-EICP differs from the lower bound by at most one.  
%We also identify a case for which the normalized rate, which is defined as the number of bits transmitted by all the users together normalized by the total bits of all the messages, achieved in our scheme for the task based CS-EICP coincides with the normalized rate achieved for the CS-EICP.
\end{abstract}
%%%%%%%%%%%%%%%%%%%%%%
\section{Introduction}

Index coding problem (ICP) is a canonical problem in network information theory, that provides a simple yet rich model for  several important engineering
problems in network communication, such as content broadcasting, peer-to-peer
communication, distributed caching, device-to-device relaying, distributed storage, and interference management \cite{birkCol2016ICP,ji2014caching,maleki2014index,mazumdar2014duality,luo2019coded}. The authors of \cite{potter2019eic} introduced a variant of ICP, called embedded index coding problem (EICP), where each node can be both sender and user at the same time. This problem is motivated by applications in distributed computation and distributed storage. It is a special case of multi-sender ICP \cite{li2018cooperative,arunachala2019optimal,Ong2019MultiSenderRank},  where the set of  users and senders are the same. It has got application in vehicular ad-hoc networks (VANETs) which have gained popularity with their importance in intelligent transport systems \cite{mac2018v2v}. In \cite{jessy2020v2v}, scalar linear index coding techniques have been applied to reduce the number of transmissions required for data exchange during the Vehicle to Vehicle (V2V) communication phase which is an integral part of collaborative message dissemination in VANETs.

EICP  consists of a set of users where each user already has a subset of messages and demands another subset of messages.  Each user is fully aware of the content available at all other users and  can communicate to all its peers through an error-free broadcast channel. The goal is to minimize the number of bits transmitted by all the users such that each user retrieves whatever they have demanded. There are no separate senders involved in this setting. Some results establishing relationships between single sender (centralized) index coding and EICP have been provided in \cite{potter2019eic}. In particular, it is shown that, the  optimal code length for an EICP is only a factor of two  worse than the optimal code length for a single sender index coding problem with the same setting. A heuristic algorithm has also been proposed for EICP. In \cite{anjana2020eic}, for EICP, a notion of side-information matrix was introduced. The length of an optimal scalar linear index code was derived to be equal to the min-rank of the side-information matrix. 

%A new notion of  task based EICP was introduced in \cite{potter2019eic} and was further studied in \cite{haviv2019taskbasedEIC}.  In task based EICP, each user can decode, the messages it has demanded, from the transmissions provided by only one other user. This is motivated by the study of instantly decodable network code \cite{le2013instantdecoding}. Task based solutions are more robust to delays and failures \cite{potter2019eic}.  In \cite{haviv2019taskbasedEIC}, it is shown that the minimum length of a linear task based embedded index code is at most quadratic in the minimum length of a linear index code in the centralized setting. 

In this paper, we consider a specific class of embedded index coding problem, defined as \textit{Consecutive and Symmetric Embedded Index Coding Problem (CS-EICP)}. 
%We also study the task based solutions to CS-EICP - \textit{Task Based CS-EICP}. 
We assume that the cardinality of the side-information is same for all the users. The {\it normalized rate} is defined as the total number of bits transmitted by all the users together normalized by the total bits of all the messages. 

 In \cite{potter2019eic}, the proposed heuristic algorithm for EICP involves calculating min-rank of a graph, by searching over all possible fitting matrices, which is computationally complex. 
 In \cite{anjana2020eic}, the CS-EICP was studied as `one-sided neighboring side-information problem'. The authors had characterized the length of the optimal scalar linear index code for CS-EICP to be $N-s+1$, where $N$ represents the number of users as well as messages and $s$ represents the cardinality of side-information available at each user. A scalar linear code achieving this length was also constructed. Hence the  normalized rate is $\frac{N-s+1}{N}$.
In this paper, we provide an explicit code construction for the CS-EICP by appropriately invoking sub-packetization of the messages. The normalized rate achieved in our scheme is   $\frac{1}{\left \lceil \frac{s}{N-s} \right \rceil}$, if $s> \frac{N}{2}$ and 
$\frac{	\left \lceil \frac{N-s}{s-1} \right \rceil }{1+ 	\left \lceil \frac{N-s}{s-1} \right \rceil }$,  if $s\leq \frac{N}{2}$. For certain ranges of values of $s$, we prove that it is less than $\frac{N-s+1}{N}$.  

 One of the special cases of EICP is when the users demand all the messages which are not in the side-information. This special case was studied as Cooperative Data Exchange (CDE) problem  in \cite{sadegi2010CDE}, where there is a set of $M$ messages and $N$ users which demand the
whole message set. Each user already has a subset of the
messages available as side-information. Upper and lower bounds on the minimum
number of transmissions are provided in \cite{sadegi2010CDE}. For the
case when all the users have the same number of messages, i.e. $s$, as
side-information, the lower bound on the number of transmissions required is $ M - s+ 1$, i.e., the normalized rate is lower bounded by $\frac{M-s+1}{M}$. If our scheme is specialized to CDE problem, then the normalized rate achieved in our scheme is $\frac{1}{\left \lceil \frac{s}{N-s} \right \rceil}$, if $s> \frac{N}{2}$ and 
$\frac{	\left \lceil \frac{N-s}{s-1} \right \rceil }{1+ 	\left \lceil \frac{N-s}{s-1} \right \rceil }$,  if $s\leq \frac{N}{2}$. Here also, for some cases, we prove that it is less than $\frac{M-s+1}{M}$.

\subsection{Vector linear code and sub-packetization scheme.}

An index coding scheme is said to be linear if the transmitted index code symbols are linear combinations of the messages. A scalar linear code uses only one instant of the $ M $ message symbols to obtain the index code symbols whereas a vector linear code uses multiple instants of $ M $ messages to obtain the index code symbols. For example, if the sender uses two instants of $ M $ messages and sends $ n $ linear index code symbols, then it means that $ n $ linear combinations of $ 2M $ messages are broadcast and the code is a vector linear code. 

In sub-packetization scheme that we introduce in this paper for index coding problems, we do not use multiple instants of messages. We use only one instant of the $ M $ message symbols while we split each message of size $d$ bits into $z$ blocks.  We assume that $d$ is sufficiently large such that this splitting of message into $z$ blocks of equal sizes is possible. The size of each block is $d_1=\frac{d}{z}$ bits and each block is assumed to be from a finite field $\mathbb{F}_{2^{d_1}}.$ The coded symbols transmitted are a linear combination of these blocks rather than the linear combination of the entire messages.  Sub-packetization is extensively used and studied in the coded caching literature.

\subsection{Our Contributions}
 The contributions of this paper is summarized as follows.

 \begin{itemize}
 	\item We introduce the idea of sub-packetization in index coding problems to provide code construction for a special class of EICP, namely \textit{Consecutive and Symmetric Embedded Index Coding Problem (CS-EICP)}. 
 	\item We show that, for CS-EICP, the normalized rate achieved in our scheme is  $\frac{1}{\left \lceil \frac{s}{N-s} \right \rceil}$, if $s> \frac{N}{2}$ and $\frac{	\left \lceil \frac{N-s}{s-1} \right \rceil }{1+ 	\left \lceil \frac{N-s}{s-1} \right \rceil }$,  if $s\leq \frac{N}{2}$. We prove that, when $ (s- 1)$ divides $ (N-1)$ or $(N-s)$ divides $(N-1)$ or $s > \frac{2N+1-\sqrt{4N+1}}{2}$, this is less than the normalized rate $\frac{N-s+1}{N}$ achieved in \cite{anjana2020eic} using scalar linear code, where $N$ represents the number of users as well as messages and $s$ represents the cardinality of side-information available at each user.
 	\item One of the special cases of EICP is when it is specialized to cooperative data exchange problem. For such cases also, the normalized rate achieved in our case is $\frac{1}{\left \lceil \frac{s}{N-s} \right \rceil}$, if $s> \frac{N}{2}$ and 
 	$\frac{	\left \lceil \frac{N-s}{s-1} \right \rceil }{1+ 	\left \lceil \frac{N-s}{s-1} \right \rceil }$,  if $s \leq \frac{N}{2}$. 
 	We prove that, when $ (s- 1)$ divides $ (N-1)$ or $(N-s)$ divides $(N-1)$ or $s > \frac{2N+1-\sqrt{4N+1}}{2}$, this
 	 is less than the lower bound on the normalized rate, which is $\frac{M-s+1}{M}$, for scalar linear solutions to CDE problem \cite{sadegi2010CDE}.
 \end{itemize}
The rest of the paper is organized as follows. The background and preliminaries are provided in Section \ref{sec: prob}. In Section \ref{sec: main}, we define the specific class of EICP considered in this paper, namely, \textit{Consecutive and Symmetric Embedded Index Coding Problem (CS-EICP)}. Our main result is summarized in the same section. Comparison of our results with the prior works is also done in the same section.
The proof of this result is deferred to Section \ref{sec achievability}. 
%Task based extension of CS-EICP is studied in Section \ref{sec: task}. Our result for the task based CS-EICP is discussed in Theorem \ref{thm: main result2} in the same section. The proof of Theorem \ref{thm: main result2} is provided in Sections \ref{sec TB}.
 Section \ref{conclusion} concludes this paper.

{\it Notations:}  The finite field with $q$ elements is denoted  by $\mathbb{F}_q.$ The set of all integers is denoted by $\mathbb{Z}.$ $[n]$ represents the set $\{1,2, \ldots , n\}$. $[a,b]$ represents the set $\{ a, a+1, \ldots, b \}$, and $(a,b]$ represents the set $\{ a+1, \ldots, b \}$.  The bit wise exclusive OR (XOR) operation is denoted by $\oplus.$  $\lfloor x \rfloor$ denotes the largest integer smaller or equal than $x$. $\lceil x \rceil$ denotes the smallest integer greater than or equal than $x$. All the message indices are taken modulo $M$ while the user indices are taken modulo $N$. $a|b$ implies $a$ divides $b$ and $a\ndiv b$ implies $a$ does not divide $b$, for integers $a$ and $b$.

%%%%%%%%%%%%%%%%%%%%%%%%%%%%%%%%%%%%%%%%%%%%%%%%%%%%%%%%%%%%%
\section{Background and Preliminaries} \label{sec: prob}

Consider a system consisting of $N$ users $$S=\{S_0,S_1, \ldots,S_{N-1}\}$$ and  $M$ messages of $d$ bits each, $$X=\{x_0,x_1,...,x_{M-1}\}, x_l \in \mathbb{F}_{2^{d}} , \forall l \in [0,M-1].$$
Let $K_j \subseteq X$ represent the subset of messages held by the user $S_j$  and $W_j \subseteq X$ represent the subset of messages demanded by the user $S_j$, $j \in [0,N-1]$. We assume that $\cup_{j \in [0,N-1]}K_j=X$.
Each user $S_j$ broadcasts a set of $y_j$ coded symbols each of size $d_1=\frac{d}{z}$ bits, for some $z \in \mathbb{Z}$. Let $\mathcal{Y}_j, j \in [0,N-1],$ represent the set of all coded symbols transmitted by the user $S_j$,
$$\mathcal{Y}_j = \cup_{i=1}^{y_j} Y_j^{i} ,  \text{ : } Y_j^{i} \in \mathbb{F}_{2^{d_1}} ,$$
where $Y_j^i, i \in [y_j ], $ represents the  $i^{th}$ coded symbol of length $d_1$ bits, transmitted by the user $S_j$.

The \textit{embedded index coding problem (EICP)} \cite{potter2019eic} is to minimize the number of bits broadcast
by all users such that each user gets all the messages they have demanded, from the messages available with them and the coded symbols broadcast by the other users. That is, to minimize the \textit{normalized rate}, which is defined as the total number of bits broadcast by all the users together normalized by the total bits of all the messages. 

The decoding function, for embedded index coding problem, associated with some user $S_j$, is of the form $$D_j :  \{\cup_{i \in \{[0,N-1] \backslash j\}}\mathbb{F}_{2^{y_i d_1}} ,  \mathbb{F}_{2^{|K_j|d}}\} \rightarrow \mathbb{F}_{2^{|W_j|d}}.$$ 

%The \textit{task based embedded index coding problem (TBEICP)} \cite{potter2019eic} is to minimize the number of bits broadcast
%by all users such that each user gets all the messages they have demanded, from the messages available with them and the coded symbols broadcast by any one user. The restriction in task based embedded index coding problem compared to embedded index coding problem is that each user can use only transmissions from any one user to decode the demanded messages.
%
%The decoding function, for task based embedded index coding problem, associated with some user $S_j$, is of the form $$D_j :  \{\mathbb{F}_{2^{y_k d_1}} ,  \mathbb{F}_{2^{|K_j|d}}\} \rightarrow \mathbb{F}_{2^{|W_j|d}},$$ 
%for some $k \in [0,N-1] \backslash j.$ 
%%%%%%%%%%%%%%%%%%%%%%%%%%%%%%%%%%%%%%%%%%%%%%%%%%%%%%%%%
\section{Consecutive and Symmetric Embedded Index Coding Problem} \label{sec: main}
In this section, we define the specific class of EICP considered in this paper, in Definition \ref{def: CS-EICP}. We  summarize our key result subsequently in Theorem \ref{thm: main result}. The proof of Theorem \ref{thm: main result}  is provided in Section \ref{sec achievability}. We compare our results with that in \cite{potter2019eic,anjana2020eic} and \cite{sadegi2010CDE}. We also illustrate our results using some examples. 

\begin{defn} \label{def: CS-EICP}
	\textbf{Consecutive and Symmetric Embedded Index Coding Problem (CS-EICP)}: 
	An EICP is said to be Consecutive and Symmetric Embedded Index Coding Problem if the	side-information of each user $S_j, j\in[0,N-1],$ can be expressed as \\ $K_j= \{x_{(j+a) \text{ mod } M},x_{(j+a+1) \text{ mod } M}, ... ,x_{(j+a+s-1) \text{ mod } M}\} $, for some $a \in [0,M-1],s\in [1,M] $.
\end{defn}

%%%%%%%%%
\subsection{Main Result}
Without loss of generality, let the side-information set of each user $S_j, j \in [0,N-1]$, for CS-EICP, be $ K_j=\{x_{j},x_{(j+1) \text{ mod } M},\ldots , x_{(j+s-1) \text{ mod } M}\},$ for some $s \in [1,M]$. 

\begin{thm} \label{thm: main result}
 For any CS-EICP, with $M=N$, $s \in [2,N-1]$, and demand set of each user $S_j, j\in[0,N-1],$ expressed as $W_j \subseteq X \backslash K_j $,  the following normalized rate is achievable by using sub-packetization:
	\begin{equation}
	\mathcal{C}(s)=
	\begin{cases}
	 \frac{1}{\left \lceil \frac{s}{N-s} \right \rceil} , & \text{ if $s >\frac{N}{2}$}. \\
	 \frac{	\left \lceil \frac{N-s}{s-1} \right \rceil }{1+ 	\left \lceil \frac{N-s}{s-1} \right \rceil } , & \text{ otherwise}. 
	\end{cases}
	\end{equation}	
\end{thm}

%For any CS-EICP, the demand set is said to be neighboring, consecutive and symmetric, if the demand set of each user $S_j, j\in[0,N-1],$ can be expressed as  $W_j =\{x_{j+s},x_{j+s+1}, \ldots,x_{j+s+t}  \}$, for some $t \in [0,N-s-1]$.
%\begin{thm} \label{thm: bounds}
% For any CS-EICP, with $M=N$, $s > \frac{N}{2}$, and demand set of each user $S_j, j\in[0,N-1],$ expressed as $W_j =\{x_{j+s},x_{j+s+1}, \ldots,x_{j+s+t}  \}$ for some $t \in [0,N-s-1]$ (i.e., the demand set is neighboring, consecutive and symmetric), a lower bound on the normalized rate for any scheme using sub-packetization is given by 
%		\begin{equation}
%		\mathcal{C}_{lb}(s)=
%		 \frac{1}{N}\left \lfloor \frac{N}{\left \lceil \frac{s}{N-s} \right \rceil} \right \rfloor. 
%		\end{equation}
%\end{thm}

%\begin{remark}
%	It can be noted that if $s > \frac{N}{2}$, when $\left \lceil \frac{s}{N-s} \right \rceil$ divides $N$, the normalized rate achieved in Theorem \ref{thm: main result} meets the lower bound given in Theorem \ref{thm: bounds}, if the demand set is neighboring, consecutive and symmetric.
%	We identify one such case of CS-EICP in Theorem \ref{thm: special case} for which  $\left \lceil \frac{s}{N-s} \right \rceil$ divides $N$.
%\end{remark}

\subsection{Comparison with the results in \cite{potter2019eic} and \cite{anjana2020eic}}
 In \cite{potter2019eic}, a heuristic algorithm, which provides a scalar linear solution for the EICP, had been provided which involves calculating computationally complex min-rank of a graph.  In \cite{anjana2020eic}, a scalar linear code achieving the length $N-s+1$ was constructed explicitly in contrast to the computationally complex algorithm presented in \cite{potter2019eic} to find a scalar linear solution. We prove in Theorem \ref{thm3} that for some range of values of $s$, the normalized rate achieved in our scheme, as in Theorem \ref{thm: main result}, using sub-packetization is lower than the normalized rate achieved in \cite{anjana2020eic} .
%For CS-EICP, the lower bound on the normalized rate for scalar linear index code, was proposed to be $\frac{N-s+1}{N}$ in \cite{anjana2020eic} while we obtain a new lower bound $\frac{1}{N}\left \lfloor \frac{N}{\left \lceil \frac{s}{N-s} \right \rceil} \right \rfloor $ as in Theorem \ref{thm: bounds}, if $s> \frac{N}{2}$, for any scheme using sub-packetization.
\begin{thm}
	\label{thm3}
	For any CS-EICP, with $M=N$, and demand set of each user $S_j, j\in[0,N-1],$ expressed as $W_j \subseteq X \backslash K_j $, when $ (s- 1)| (N-1)$ or $(N-s)|(N-1)$ or $ \frac{2N+1-\sqrt{4N+1}}{2} <s<N$, the normalized rate achieved in our scheme, as in Theorem \ref{thm: main result}, using sub-packetization is lower than the normalized rate $\frac{N-s+1}{N}$ achieved in \cite{anjana2020eic} using scalar linear index code.
\end{thm}
\begin{proof}
	The length of the scalar linear code constructed in \cite{anjana2020eic} is $ N - s+ 1$, i.e., the normalized rate is $\frac{N-s+1}{N}.$ In our scheme, the normalized rate as described in Theorem \ref{thm: main result} is achievable. 
	
	{\it Case 1:} When $ (s- 1)| (N-1)$.
	
	For this case, if $s \leq \frac{N}{2}$, the normalized rate achieved in our scheme is
	\begin{align*}
	\mathcal{C}(s) &=\frac{	\left \lceil \frac{N-s}{s-1} \right \rceil }{1+ 	\left \lceil \frac{N-s}{s-1} \right \rceil }
	=\frac{	\left \lceil \frac{N-1}{s-1}-1 \right \rceil }{1+ 	\left \lceil \frac{N-1}{s-1}-1 \right \rceil  }
	\\&=\frac{ \frac{N-1}{s-1}-1 }{1+ \frac{N-1}{s-1}-1}
	=\frac{ \frac{N-s}{s-1}  }{	 \frac{N-1}{s-1}  }=\frac{N-s}{N-1}.
	\end{align*}
	Also, we have
		\begin{align*}
		N &> N-s+1\\ \Rightarrow
		N +N(N-s)&> N-s+1+N(N-s)\\
	\Rightarrow	(N-s+1)(N-1)&>N(N-s)\\
	\Rightarrow	\frac{N-s+1}{N}&>\frac{N-s}{N-1}.
		\end{align*}	
	If $s > \frac{N}{2}$, {\it Case 1} is true only when $s=\frac{N+1}{2}$, where $N$ is odd. For such cases, the normalized rate achieved in our scheme is
	\begin{align*}
	\mathcal{C}(s) = \frac{1}{\left \lceil \frac{s}{N-s} \right \rceil}= \frac{1}{\left \lceil \frac{\frac{N+1}{2}}{N-\frac{N+1}{2}} \right \rceil}=\frac{1}{\left \lceil \frac{N+1}{N-1} \right \rceil} =\frac{1}{2}.
	\end{align*}
	If $s=\frac{N+1}{2}$, where $N$ is odd, 
	\begin{align*}
		\frac{N-s+1}{N} =\frac{\frac{N-1}{2}+1}{N} =\frac{N+1}{2N}=\left(\frac{1}{2} +\frac{1}{N}\right) \left (>\frac{1}{2} \right)
	\end{align*}
	
		{\it Case 2:} When $(N-s)|(N-1)$.
		
			For this case the normalized rate achieved in our scheme is
		\begin{align*}
		\mathcal{C}(s) = \frac{1}{\left \lceil \frac{s}{N-s} \right \rceil}= \frac{1}{\left \lceil \frac{N}{N-s}-1 \right \rceil}=\frac{1}{ \frac{N-1}{N-s}+1-1}=\frac{N-s}{ N-1}.
		\end{align*}
	Hence, $\frac{N-s+1}{N} >\frac{N-s}{N-1}$, as it is proved in {\it Case 1}.
	
	{\it Case 3:} When $ \frac{2N+1-\sqrt{4N+1}}{2} <s<N$.
	
	We observe that $\frac{2N+1-\sqrt{4N+1}}{2} \geq \frac{N}{2}$ since,
	\begin{align*}
	N^2  &\geq 2N\\
	\Rightarrow N^2 +2N+1 &\geq 4N+1
	\Rightarrow (N+1)^2 \geq 4N+1\\
	\Rightarrow N+1 &\geq \sqrt{(4N+1)}\\
	\Rightarrow 2N+1- \sqrt{(4N+1)} &\geq N\\
	\Rightarrow \frac{2N+1- \sqrt{(4N+1)}}{2} &\geq \frac{N}{2}.
	\end{align*}
	Let us assume that for this case, the length of the scalar linear code as in \cite{anjana2020eic} is less than or equal to that achieved in our scheme, i.e., $ \frac{N-s+1}{N} \leq \frac{1}{\left \lceil \frac{s}{N-s} \right \rceil }$. Then we have,
		\begin{align}
		\frac{N-s+1}{N} &\leq \frac{1}{\left \lceil \frac{s}{N-s} \right \rceil } \\
		\Rightarrow \frac{N}{N-s+1} &\geq \left \lceil \frac{s}{N-s} \right \rceil \\
		\Rightarrow \frac{N}{N-s+1} &\geq \frac{s}{N-s} \nonumber\\
		\Rightarrow 	N(N-s) &\geq  s(N-s)+s  \\
		\Rightarrow 	(N-s)^2&\geq s \\
		\Rightarrow N^2 -2Ns+s^2 &\geq s  \\
		\Rightarrow s^2 -(2N+1)s+N^2 &\geq 0 
		\end{align}
		\begin{align}
		\Rightarrow \left ( s-\left ( \frac{2N+1+\sqrt{4N+1}}{2} \right )  \right )\times \\
		\left ( s-\left ( \frac{2N+1-\sqrt{4N+1}}{2} \right )  \right ) &\geq 0 \nonumber\\
		\Rightarrow (s-s_1)(s-s_2) &\geq 0, \label{eq one}
		\end{align}
\noindent where $s_1=\frac{2N+1+\sqrt{4N+1}}{2}$ and $s_2=\frac{2N+1-\sqrt{4N+1}}{2}$.
	Eq. (\ref{eq one}) implies either $s>s_1$ or $ s<s_2.$
	We observe that the value of $s$ cannot be greater than $s_1$ (since $s_1 > N$), which implies $s<s_1$. So, the only possible solution to Eq. (\ref{eq one}) is $s<s_2, $ which contradicts our assumption that $s>	\frac{2N+1-\sqrt{4N+1}}{2}$. Hence $ \frac{N-s+1}{N} >  \frac{1}{\left \lceil \frac{s}{N-s} \right \rceil }$.
	
	Therefore, we are able to achieve a normalized rate lower than the normalized rate achieved in \cite{anjana2020eic} when $ (s- 1)| (N-1)$ or $(N-s)|(N-1)$ or $s > \frac{2N+1-\sqrt{4N+1}}{2}$. 
\end{proof}
%%%%%%%%%%%%%%%%%%%%
%\begin{cor}
%	The range of $s$ covered in Theorem \ref{thm: main result} is the union of two disjoint sets $\{s:  (s- 1) \text{ divides } (N-1)\} $ and $ \left (\frac{2N+1-\sqrt{4N+1}}{2}, N-1 \right ]$.
%\end{cor}
%\begin{proof}
%	$ (s- 1)$ divides $ (N-1)$ implies $s \leq \frac{N}{2}$. 
%\end{proof}
	\begin{remark}
		For those ranges of values of $s$ which are not discussed in Theorem \ref{thm3}, i.e., when $ (s- 1) \ndiv (N-1)$, $(N-s) \ndiv (N-1)$ and $ \frac{N}{2} < s \leq\frac{2N+1-\sqrt{4N+1}}{2}$, we conjecture that the normalized rate achieved in our scheme, as in Theorem \ref{thm: main result}, using the idea of sub-packetization is lower than the normalized rate $\frac{N-s+1}{N}$ achieved in \cite{anjana2020eic} using scalar linear index code.
		\end{remark}
%%%%%%%%%%%%%%%%%%%%%%%%%%%%%%
\begin{remark}
	One of the special cases of EICP, when the users demand all the messages which are not available as side-information, was studied as Cooperative Data Exchange (CDE) problem  in \cite{sadegi2010CDE}. A lower bound on the minimum
	number of transmissions, provided in \cite{sadegi2010CDE} for the case when all the users have the same number of messages, i.e. $s$, as side-information, is $ M - s+ 1$, i.e., the normalized rate is lower bounded by $\frac{M-s+1}{M}$. If our scheme is specialized to CDE problem, when $ (s- 1)| (N-1)$ or $(N-s)|(N-1)$ or $ \frac{2N+1-\sqrt{4N+1}}{2} <s<N$,  the normalized rate achieved in our scheme, as in Theorem \ref{thm: main result}, using sub-packetization is lower than the lower bound on the normalized rate provided in \cite{sadegi2010CDE} (as proved in Theorem \ref{thm3}).
\end{remark}
The following three examples illustrate Theorem \ref{thm: main result} and also the idea of sub-packetization that is invoked in the proof.
%%%%%%%%%%%%%%%%%%%%%%%%%%%%%%%%%%%%%%%
\begin{example} \label{exp T1C1}
	Let $N=5, M=5, s=3$. Thus  we have five messages $\{x_0,x_1,x_2,x_3,x_4\}$, each of size $d$ bits,  and five users $\{S_0,S_1,S_2,S_3,S_4\}$. Let the  side-information set and the demand set corresponding to each user $S_j, j \in [0,4],$ be $ K_{j} = \{x_j,x_{(j+1) \text{ mod } 5},x_{(j+2) \text{ mod } 5}\}$ and $W_j =\{x_{(j+3) \text{ mod } 5}\}$ respectively.
		\begin{align*}
	K_{0} &= \{x_0,x_1,x_2\}&
	K_{1} &= \{x_1,x_2,x_3\}&
	K_{2} &= \{x_2,x_3,x_4\}\\
	K_{3} &= \{x_3,x_4,x_0\}&
	K_{4} &= \{x_4,x_0,x_1\}
	\end{align*}
	\begin{align*}
	W_{0}&=\{x_3\}&
	W_{1}&=\{x_4\}&
	W_{2}&=\{x_0\}\\
	W_{3}&=\{x_1\}&
	W_{4}&=\{x_2\}
	\end{align*}
	We split each message into two disjoint blocks each of size $\frac{d}{2}$ bits, i.e.,
	\begin{align*}
	x_0 &=\{x_0^0,x_0^1\}&
	x_1 &=\{x_1^0,x_1^1\}&
	x_2 &=\{x_2^0,x_2^1\}\\
	x_3 &=\{x_3^0,x_3^1\}&
	x_4 &=\{x_4^0,x_4^1\}
	\end{align*}	
	The coded symbols transmitted are linear combinations of these blocks. Each user $S_{h}, h \in [0,4],$ transmits one coded symbol $Y_h$ which includes $2$ messages taken at an interval of $2$. The $0^{th}$ block of the first message is taken while the $1^{st}$ block of the second message is taken. That is,
	for each $h \in [0,4]$, the user $S_h$ transmits $Y_h = x_{h}^0 \oplus x_{(h+2) \text{ mod } 5}^1$. The transmitted coded symbols are 
	 \begin{align*}
	Y_0 &= x_0^0 \oplus x_{2}^1&
	Y_1 &= x_1^0 \oplus x_{3}^1&
	Y_2 &= x_2^0 \oplus x_{4}^1\\
	Y_3 &= x_3^0 \oplus x_{0}^1&
	Y_4 &= x_4^0 \oplus x_{1}^1.
	\end{align*}	
	Now, each user $S_j$ needs to retrieve the demanded message $x_{(j+3) \text{ mod } 5}.$  Let us first consider the user $S_0.$  The user $S_0 $ retrieves $x_{3}^{0}$ from $Y_3$ since $x_{0}$ is available as side-information while it retrieves $x_{3}^{1}$ from $Y_1$. The user $S_0$ has decoded the  message $x_3$ since it has retrieved all the blocks corresponding to the message $x_3$. Similarly all other users can decode their demanded message.
	Table \ref{tab1 2} illustrates the coded symbols transmitted by each user and the coded symbols from which each user retrieves all the blocks corresponding to the demanded message. It can be noted from Table \ref{tab1 2} that each user transmits $\frac{d}{2}$ bits owing to a normalized rate of $\frac{1}{2}$.  The minimum number of bits required to transmit is $3d$ bits in  \cite{potter2019eic,anjana2020eic} while we were able to reduce it to $2.5d$ bits by utilizing the sub-packetization.
\end{example}
%%%%%%%%%%%%%%%%%%%%%%%%%%%%%%

\begin{table*}
	\begin{center}
		\begin{tabular}{ |p{1cm}| p{2.5cm}|p{2.5cm}|p{3cm} | p{4cm} |}
			\hline
			
			\textit{User $S_i$}& \textit{Coded symbols transmitted by $S_i$}& \textit{Message  demanded by $S_i$: $x_j$} & \textit{Message blocks corresponding to the message $x_j$} & \textit{Coded symbols from which the message blocks are decoded by $S_i$}   \\ \hline 
			\hline
			$S_0 $& $Y_0=x_0^0 \oplus x_2^1$  & $x_3$	& $x_3^0$ & $Y_3$  \\ \cline{4-5}
			& &&$x_3^1$ & $Y_1$  \\ \hline
			
			$S_1 $ & $Y_1=x_1^0 \oplus x_3^1$&  $x_4$	& $x_4^0$ & $Y_4$  \\ \cline{4-5}
			& &&$x_4^1$ & $Y_2$  \\ \hline
			
			$S_2 $ & $Y_2=x_2^0 \oplus x_4^1$& $x_0$	& $x_0^0$ & $Y_0$  \\ \cline{4-5}
			& &&$x_0^1$ & $Y_3$  \\ \hline
			$S_3 $ & $Y_3=x_3^0 \oplus x_0^1$&  $x_1$	& $x_1^0$ & $Y_1$  \\ \cline{4-5}
			& &&$x_1^1$ & $Y_4$  \\ \hline
			
			$S_4 $ & $Y_4=x_4^0 \oplus x_1^1$ & $x_2$	& $x_2^0$ & $Y_2$  \\ \cline{4-5}
			& &&$x_2^1$ & $Y_0$  \\ \hline
			
		\end{tabular}
	\end{center}
	\caption{Table that illustrates the decoding done by each user in Example \ref{exp T1C1}}
	\label{tab1 2}
\end{table*}
%%%%%%%%%%%%%%%%%%%%%%%%%%%%%%%%%%%%%%%%%%%%%%%%%%%%
\begin{example} \label{exp T1C2}
	Let us take an  example for the case when $s \leq \frac{N}{2}$ in Theorem \ref{thm: main result}. Let $N=M=4,s=2$ and the set of all messages and users be $\{x_0,x_1,x_2,x_3\}$ and $\{S_0,S_1,S_2,S_3\}$ respectively. Let the side-information set and the demand set corresponding to each user $S_j, j \in [0,3]$ be $ K_{j} = \{x_j,x_{(j+1) \text{ mod } 4}\}$ and $W_j =\{x_{(j+2) \text{ mod } 4}\}$ respectively.
		\begin{align*}
	K_{0} &= \{x_0,x_1\}&
	K_{1} &= \{x_1,x_2\}
\\
	K_{2} &= \{x_2,x_3\}&
	K_{3} &= \{x_3,x_0\}
	\end{align*}
	\begin{align*}
	W_{0}&=\{x_2\}&
	W_{1}&=\{x_3\}&
	W_{2}&=\{x_0\}&
	W_{3}&=\{x_1\}
	\end{align*}
	 We split each message into three blocks of equal sizes, $x_j =\{x_j^0,x_j^1,x_j^2\}, j \in[0,3]$.
	 	\begin{align*}
	 x_0 &=\{x_0^0,x_0^1,x_0^2\}&
	 x_1 &=\{x_1^0,x_1^1,x_1^2\}\\
	 x_2 &=\{x_2^0,x_2^1,x_2^2\}&
	 x_3 &=\{x_3^0,x_3^1,x_3^2\}
	 \end{align*}
	 	The coded symbols transmitted are linear combinations of these blocks. For this case, the coded symbols are obtained in $4$ iterations. For each iteration $h \in [0,3] $,  the users $S_{h}$ and $S_{(h+1) \text{ mod } 4} $ are involved in the transmissions, where the coded symbols obtained by each user is by taking the first and the last messages available at each user. The user $S_{h}$ transmits one coded symbol $Y_h^0$ where the $0^{th} $ block of the first message and the $1^{st}$ block of the last message available as side information are taken to be included in the coded symbol. 
	The user $S_{(h+1) \text{ mod } 4} $ transmits one coded symbol $Y_h^1$, where the $1^{st} $ block of the first message and the $2^{nd}$ block of the last message available as side information are taken to be included in the coded symbol. That is,
	  for each $h \in [0,3], i \in [0,1], $  the user $S_{(h+i) \text{ mod } 4}$ transmits $Y_h^i = x_{(h+i) \text{ mod } 4}^i \oplus x_{(h+i+1) \text{ mod } 4}^{i+1}$. The coded symbols transmitted are given below.
	   \begin{align*}
	  Y_0^0 &= x_0^0 \oplus x_{1}^1&
	  Y_0^1 &= x_1^1 \oplus x_{2}^2&
	  Y_1^0 &= x_1^0 \oplus x_{2}^1\\
	  Y_1^1 &= x_2^1 \oplus x_{3}^2&
	  Y_2^0 &= x_2^0 \oplus x_{3}^1&
	  Y_2^1 &= x_3^1 \oplus x_{0}^2\\
	  Y_3^0 &= x_3^0 \oplus x_{0}^1&
	  Y_3^1 &= x_0^1 \oplus x_{1}^2
	  \end{align*}
	  Now, each user $S_j$ needs to retrieve the demanded message $x_{(j+2) \text{ mod } 4}.$  Let us first consider the user $S_0.$  The user $S_0 $ retrieves $x_{2}^{0}$ from $Y_2^0 \oplus Y_2^1 =x_2^0 \oplus x_0^2$ since $x_{0}$ is available as side-information while it retrieves $x_{2}^{1}$ and $x_{2}^{2 } $ from $Y_1^0$ and $Y_0^1$ respectively. The user $S_0$ has decoded the  message $x_2$ since it has retrieved all the blocks corresponding to the message $x_2$. Similarly all other users can decode their demanded message.
	  
	   Table \ref{tab2 2} illustrates the coded symbols transmitted by each user and the coded symbols from which each user retrieves all the blocks corresponding to the demanded message. It can be verified from Table \ref{tab2 2}  that the normalized rate in this particular example is $\frac{2}{3}$, since each user transmits $\frac{2d}{3}$ bits. Here the total number of bits transmitted by all the users together is $\frac{8d}{3}$ bits which is less than $3d$ bits required to transmit in  \cite{potter2019eic,anjana2020eic}.
\end{example}
%%%%%%%%%%%%%%%%%%%%%%%%%%%%%%%%%%%%%%%

\begin{table*}
	\begin{center}
		\begin{tabular}{ |p{1cm}| p{2.5cm}|p{2.5cm} | p{3cm} | p{4.4cm} |}
			\hline
			
			\textit{User $S_i$}&\textit{Coded symbols transmitted by $S_i$} 	&\textit{Message demanded by $S_i$: $x_j$}& \textit{Message blocks corresponding to the message $x_j$} & \textit{Coded symbols from which the message blocks are decoded by $S_i$}   \\ \hline \hline
			
			&$Y_0^0=x_0^0 \oplus x_1^1$& &$x_2^0$& $Y_2^0 \oplus Y_2^1 $  \\ \cline{4-5}
			$S_0$&$Y_3^1=x_0^1 \oplus x_1^2$ &$x_2$& $x_2^1$ & $Y_1^0$  \\  \cline{4-5}
			& &&$x_2^2$ & $Y_0^1$  \\ \hline
			
			& $Y_1^0=x_1^0 \oplus x_2^1$	& & $x_3^0$& $Y_3^0 \oplus Y_3^1 $  \\  \cline{4-5}
			$S_1$&  $Y_0^1=x_1^1 \oplus x_2^2$&$x_3$& $x_3^1$ & $Y_2^0$  \\  \cline{4-5}
			&& &$x_3^2$ & $Y_1^1$  \\ \hline
			
			&  $Y_2^0=x_2^0 \oplus x_3^1$&	& $x_0^0$ & $Y_0^0 \oplus Y_0^1 $  \\  \cline{4-5}
			$S_2$	& $Y_1^1=x_2^1 \oplus x_3^2$ &$x_0$& $x_0^1$ & $Y_3^0$  \\  \cline{4-5}
			& &&$x_0^2$ & $Y_2^1$  \\ \hline
			
			&	$Y_3^0=x_3^0 \oplus x_0^1$ &	& $x_1^0$ & $Y_1^0 \oplus Y_1^1 $  \\  \cline{4-5}
			$S_3$&  $Y_2^1=x_3^1 \oplus x_0^2$ &$x_1$& $x_1^1$ & $Y_0^0$  \\  \cline{4-5}
			&& &$x_1^2$ & $Y_3^1$  \\ \hline
				
		\end{tabular}
	\end{center}
	\caption{Table that illustrates the decoding done by each user in Example \ref{exp T1C2}}
	\label{tab2 2}
\end{table*}
%%%%%%%%%%%%%%%%%%%%%%%%%%%%%%%%%%%%%%%%%%%%%%%
\begin{example} \label{exp T1m2}
	Consider an  example for  the case where each user demands two messages. Let $N=M=7,s=4$ and the set of all messages and users be $\{x_0,x_1,x_2,x_3,x_4,x_5,x_6\}$ and $\{S_0,S_1,S_2,S_3,S_4,S_5,S_6\}$ respectively. Let the side-information set and the demand set corresponding to each user $S_j, j \in [0,6]$ be $ K_{j} = \{x_j,x_{(j+1) \text{ mod } 7},x_{(j+2) \text{ mod } 7},x_{(j+3) \text{ mod } 7}\}$ and $W_j =\{x_{(j+4) \text{ mod } 7},x_{(j+5) \text{ mod } 7}\}$ respectively.
	\begin{align*}
	K_{0} &= \{x_0,x_1,x_2,x_3\}&
	K_{1} &= \{x_1,x_2,x_3,x_4\}\\
	K_{2} &= \{x_2,x_3,x_4,x_5\}&
	K_{3} &= \{x_3,x_4,x_5,x_6\}\\
	K_{4} &= \{x_4,x_5,x_6,x_0\}&
	K_{5} &= \{x_5,x_6,x_0,x_1\}\\
	K_{6} &= \{x_6,x_0,x_1,x_2\}
	\end{align*}
	\begin{align*}
	W_{0}&=\{x_4,x_5\}&
	W_{1}&=\{x_5,x_6\}&
	W_{2}&=\{x_6,x_0\}\\
	W_{3}&=\{x_0,x_1\}&
	W_{4}&=\{x_1,x_2\}&
	W_{5}&=\{x_2,x_3\}\\
	W_{6}&=\{x_3,x_4\}
	\end{align*} 
We split each message into two blocks, i.e.,
		\begin{align*}
	x_0 &=\{x_0^0,x_0^1\}&
	x_1 &=\{x_1^0,x_1^1\}&
	x_2 &=\{x_2^0,x_2^1\}\\
	x_3 &=\{x_3^0,x_3^1\}&
	x_4 &=\{x_4^0,x_4^1\}&
	x_5 &=\{x_5^0,x_5^1\}\\
	x_6 &=\{x_6^0,x_6^1\}
	\end{align*}
		The coded symbols transmitted are linear combinations of these blocks. Each user $S_{h}, h \in [0,6],$ transmits one coded symbol $Y_h$ which includes $2$ messages taken at an interval of $2$. The $0^{th}$ block of the first message is taken while the $1^{st}$ block of the second message is taken. That is,
	 for each $h \in [0,6]$, the user $S_h$ transmits $Y_h = x_h^0 \oplus x_{(h+2) \text{ mod } 7}^1$.
	  \begin{align*}
	 Y_0 &= x_0^0 \oplus x_{2}^1&
	 Y_1 &= x_1^0 \oplus x_{3}^1&
	 Y_2 &= x_2^0 \oplus x_{4}^1\\
	 Y_3 &= x_3^0 \oplus x_{5}^1&
	 Y_4 &= x_4^0 \oplus x_{6}^1&
	 Y_5 &= x_5^0 \oplus x_{0}^1\\
	 Y_6 &= x_6^0 \oplus x_{1}^1
	 \end{align*}
	  Now, each user $S_j$ needs to retrieve the demanded messages $x_{(j+4) \text{ mod } 7}$ and $x_{(j+5) \text{ mod } 7}$.  Let us first consider the user $S_0.$  The user $S_0 $ retrieves $x_{4}^{0}$ and $x_4^1$ from $Y_4$ and $Y_1$ respectively while it retrieves $x_{5}^{0}$ and $x_{5}^{1} $ from $Y_5$ and $Y_2$ respectively. The user $S_0$ has decoded the  messages $x_4$ and $x_5$ since it has retrieved all the blocks corresponding to those message. Similarly all other users can decode their demanded messages.
	  Table \ref{tab2 m2} illustrate the coded symbols transmitted by each user and the coded symbols from which each user retrieves all the blocks corresponding to the demanded messages.
	   The total number of bits transmitted by all the users together is $\frac{7d}{2}$ bits compared to $ 4d$ bits required in \cite{potter2019eic,anjana2020eic}. The normalized rate is $\frac{1}{2}$.
\end{example}
%%%%%%%%%%%%%%%%%%%%%%%%%%%%%%%%%%%%%%%%%%%%%
%
\begin{table*}
	\begin{center}
		\begin{tabular}{ |p{1cm}| p{3cm}|p{2.5cm} | p{3cm} | p{4.4cm} |}
			\hline
			
			\textit{User $S_i$}&\textit{Coded symbols transmitted by $S_i$} 	&\textit{Message demanded by $S_i$: $x_j$}& \textit{Message blocks corresponding to the message $x_j$} & \textit{Coded symbols from which the message blocks are decoded by $S_i$}   \\ \hline \hline
			
			$S_0 $ & $Y_0=x_0^0 \oplus x_3^1$ & $x_4$	& $x_4^0$ & $Y_4$  \\ \cline{4-5}
			& &&$x_4^1$ & $Y_1$  \\ \cline{3-5}
			& &$x_5$	& $x_5^0$ & $Y_5$  \\ \cline{4-5}
			& &&$x_5^1$ & $Y_2$  \\ \hline
			
			$S_1$& $Y_1=x_1^0 \oplus x_4^1$  &  $x_5$	& $x_5^0$ & $Y_5$  \\ \cline{4-5}
			& &&$x_5^1$ & $Y_2$ \\  \cline{3-5}
			& &$x_6$	& $x_6^0$ & $Y_6$  \\ \cline{4-5}
			& &&$x_6^1$ & $Y_3$  \\ \hline

			$S_2$& $Y_2=x_2^0 \oplus x_5^1$	& $x_6$	& $x_6^0$ & $Y_6$  \\ \cline{4-5}
			& &&$x_6^1$ & $Y_3$  \\ \cline{3-5}
			& &$x_0$	& $x_0^0$ & $Y_0$  \\ \cline{4-5}
			& &&$x_0^1$ & $Y_4$  \\ \hline
			
			$S_3 $ & $Y_3=x_3^0 \oplus x_6^1$& $x_0$	& $x_0^0$ & $Y_0$  \\ \cline{4-5}
			& &&$x_0^1$ & $Y_4$  \\ \cline{3-5}
			& &$x_1$	& $x_1^0$ & $Y_1$ \\ \cline{4-5}
			& &&$x_1^1$ & $Y_5$  \\ \hline
			
			$S_4$ & $Y_4=x_4^0 \oplus x_0^1$&  $x_1$	& $x_1^0$ & $Y_1$  \\ \cline{4-5}
			& &&$x_1^1$ & $Y_5$ \\ \cline{3-5}
			&& $x_2$	& $x_2^0$ & $Y_2$  \\ \cline{4-5}
			& &&$x_2^1$ & $Y_6$  \\ \hline
			
			$S_5 $ & $Y_5=x_5^0 \oplus x_1^1$& $x_2$	& $x_2^0$ & $Y_2$  \\ \cline{4-5}
			& &&$x_2^1$ & $Y_6$ \\ \cline{3-5}
			& &$x_3$	& $x_3^0$ & $Y_3$  \\ \cline{4-5}
			& &&$x_3^1$ & $Y_0$  \\ \hline
			
			$S_6$ & $Y_6=x_6^0 \oplus x_2^1$& $x_3$	& $x_3^0$ & $Y_3$  \\ \cline{4-5}
			& &&$x_3^1$ & $Y_0$  \\ \cline{3-5}
			&& $x_4$	& $x_4^0$ & $Y_4$  \\ \cline{4-5}
			& &&$x_4^1$ & $Y_2$  \\ \hline

		\end{tabular}
	\end{center}
	\caption{Table that illustrates the decoding done by each user in Example \ref{exp T1m2}}
	\label{tab2 m2}

\end{table*}

\section{Proof of Theorem \ref{thm: main result}} \label{sec achievability}
In this section, we prove the achievability of Theorem \ref{thm: main result} by providing a sub-packetization scheme. We split this problem into two disjoint cases depending on the value of $s$. We  construct code for the two cases separately in the coming subsections.  The proposed achievable schemes in both cases involve splitting the messages and transmitting their linear combination.

We split each message into $z$ blocks, $x_l =\{x_l^0, x_l^1, \ldots , x_l^{z-1}\}, l \in [0,N-1].$ The value of $z$ is given later in the coming subsections. We assume that $d$ is sufficiently large such that this splitting of message into $z$ blocks of equal sizes is possible. The size of each block is $d_1=\frac{d}{z}$ bits. Each block is from a finite field $\mathbb{F}_{2^{d_1}}.$ Each user transmits a linear combination of these blocks rather than the linear combination of the entire messages. All the users should be able to retrieve all the blocks corresponding to the demanded messages.

\subsection{Case A:
	$s >  \frac{N}{2} $. }
In this subsection,  we provide an achievable scheme for  Case A.

Let $z=\left \lceil \frac{s}{N-s} \right \rceil.$
%\begin{align*}
%s&=(N-s)p+q \\
%p&=
%\begin{cases*}
%p +1, & \text{ if $q \neq 0$} \\
%p, & \text{ otherwise.}
%\end{cases*}
%\end{align*}
We split each message into $z$ blocks, $x_l =\{x_l^0, x_l^1, \ldots , x_l^{z-1}\}, l \in [0,N-1].$ 
	The coded symbols transmitted are linear combinations of these blocks. 
Now, we provide the code construction.
\begin{constr} \label{con 1}
	Each user $S_j, \forall j\in [0,N-1],$  transmits one coded symbol $Y_j$, where $$Y_j =  \bigoplus_{k \in [0,z-1] } x_{(k(N-s)+j) \text{ mod } N}^{k}  $$
\end{constr}
Each user $S_{j}$ transmits one coded symbol $Y_j$ which includes $z$ messages taken at an interval of $(N-s)$. Also, $z$ different blocks of these $z$ messages are chosen, i.e., $0^{th}$ block of the first message is taken, $1^{st}$ block of the second message and so on.
Since each of the messages in $\{\cup_{k \in [0,z-1] } x_{(k(N-s)+j) \text{ mod } N} \}$ is available with the user $S_j, $ the coded symbol $Y_j $ can be transmitted by $S_j$.

We need to establish that all the users are capable of retrieving  all the demanded messages  from the coded symbols obtained by Construction \ref{con 1} and the side-information. 

\textit{Proof of Decoding:} Now, we prove that each user $S_j, j \in [0, N-1],$ can retrieve each of its  demanded message $x_l \in W_j, l \in [0,N-1]\backslash [j,(j+s-1) \text{ mod } N]$.

 It can be noted from Construction \ref{con 1} that in any coded symbol $Y_{l'}$, for some $l' \in [0,N-1],$ if we take any block of a message present in $Y_{l'}$ which is needed by some user $S_{h}, h\in [0,N-1],$ then it can safely retrieve that block from $Y_{l'}$ since all other blocks in $Y_{l'}$ are available as side-information for the user $S_{h}.$ This is since all the $z$ messages included  in $Y_{l'}$ are taken at an interval of $N-s$ and $(z-1)(N-s)<s   ($since $ z=\left \lceil \frac{s}{N-s} \right \rceil)$. Therefore, for each $i \in [0,z-1]$, the user $S_j$ can retrieve $x_l^{i}$ from $Y_{l'} $, where $l'=((l - (N-s)i) \text{ mod } N)$ as $(i+1)^{th}$ message chosen to be included in $Y_{l'} $ is $x_{l}$ and $i^{th}$ block of $x_{l}$ is chosen.
%Let
%$$D_l^{i} = Y_{l'} \bigoplus_{k \in [0,z-1] \backslash i} x_{(k(N-s)+l') \text{ mod } N}^{k},$$
%where $i \in [0,z-1], l'= (l - (N-s)i) \text{ mod } N$. 
%If the set of  messages $\cup_{k \in [0,z-1] \backslash i} x_{(k(N-s)+l') \text{ mod } N}$ is available at the user $S_j$, then it can retrieve $x_l$ from $D_l^{i},i \in [0,N-1],$ since
\begin{align*}
Y_{l'} &=  \bigoplus_{k \in [0,z-1] } x_{(k(N-s)+l') \text{ mod } N}^{k} \\
&=  x_{(i(N-s)+l') \text{ mod } N}\bigoplus_{k \in [0,z-1] \backslash i} x_{(k(N-s)+l') \text{ mod } N}^{k} \\
  &= x_{l }^{i} \underbrace{ \bigoplus_{k \in [0,z-1] \backslash i} x_{(k(N-s)+l') \text{ mod } N}^{k}}_{\text{available as side-information}} 
\end{align*}
%Now, we prove that the set of  messages $\cup_{k \in [0,z-1] \backslash i} x_{(k(N-s)+l') \text{ mod } N}$ is available at the user $S_j.$ 
%\begin{itemize}
%	\item For $k>i , i.e., k \in [i+1, z-1],a \in [1,z-1-i], $
%	\begin{align*}
%	(k(N-s)+l')  \text{ mod } M&= (i(N-s)+l' + a(N-s)) \text{ mod } N \\
%	&= (l + a(N-s) ) \text{ mod } N
%	\end{align*} 
%	\item For $k<i , i.e., k \in [0,i-1]$,
%	\begin{align*}
%	(k(N-s)+l') \text{ mod } N &= (i(N-s)+l' - b(N-s)) \text{ mod } N,&b \in [1,i]  \\
%	&=( l - b(N-s)) \text{ mod } N
%	\end{align*} 
%\end{itemize} 
% $x_{(l + a(N-s)) \text{ mod } N}, x_{(l - b(N-s)) \text{ mod } N} \in K_j, $ for any $a \in [1,z-1-i]$ and $b \in [1,i]$, since $a,b < z$. Hence, the user $S_j$ can retrieve $x_l$ from $D_{l}^{i},i \in [0,z-1].$

\subsection{Case B:
	$s \leq \frac{N}{2}  $. }
In this subsection,  we provide an achievable scheme for  Case B. Let $z=1+\left \lceil \frac{N-s}{s-1} \right \rceil.$
%\begin{align*}
%N-s&=(s-1)u +v\\
%u&=
%\begin{cases*}
%u +1, & \text{ if $v \neq 0$ } \\
%u, & \text{ otherwise.}
%\end{cases*}
%\end{align*}
We split each message into $z$ blocks, $x_l =\{x_l^0, x_l^1, \ldots , x_l^{z-1}\}, l \in [0,N-1].$ 	The coded symbols transmitted are linear combinations of these blocks. 
The code construction for this case is given below.
\begin{constr} \label{con cont 2}
	For each iteration $i \in [0,N-1]$,
	\begin{itemize}
		\item each user $S_{(k(s-1)+ i) \text{ mod } N}, k \in [0,z-2]$, transmits one coded symbol $Y_k^i$, where $$Y_k^i =  x_{(k(s-1)+i) \text{ mod } N}^k \oplus x_{((k +1) (s-1) +i) \text{ mod } N}^{k+1}  .$$
	\end{itemize}
	
\end{constr} 
The coded symbols are obtained in $N$ iterations. During each iteration $i \in [0,N-1]$, $z$ messages at an interval of $s-1$ are chosen, $\left (\bigcup_{k\in [0,z-1] }x_{(k(s-1)+ i) \text{ mod } N} \right )$, and  we make sure that $z$ different blocks of these $z$ messages are taken, i.e., $0^{th }$ block of the first message is taken, $1^{st}$ block of the second message and so on. 
And also, we choose $z-1$ users at an interval of $(s-1) $, i.e., $S_{(k(s-1)+ i) \text{ mod } N}, k \in [0,z-2]$, which are involved in the transmissions during iteration $i$, where the coded symbols obtained by each user is by taking the first and the last messages available at each user. Each user $S_{(k(s-1)+ i) \text{ mod } N}$ transmits one coded symbol $Y_k^i$ where the $k^{th} $ block of the first message $x_{(k(s-1)+ i) \text{ mod } N}$ and the $(k+1)^{th}$ block of the last message $x_{((k+1)(s-1)+ i) \text{ mod } N}$ available as side information are taken to be included in the coded symbol. 
%The subscript of coded symbols is taken modulo $z$ and the superscript is taken modulo $N$.
%Since each user $S_{(k(s-1)+i ) \text{ mod } N}$ has the messages $x_{(k(s-1)+i) \text{ mod } N}$ and $ x_{((k+1) (s-1) +i) \text{ mod } N} $ in its side-information set, it can transmit $Y_k^i$, for each $i \in [0,N-1], k \in [0,z-1].$ 

\textit{Proof of Decoding:} We need to prove that each user $S_j, j\in[0,N-1],$ can retrieve all the blocks corresponding to all the messages in $ W_j$. Let $x_l \in W_j, l \in [0,N-1]\backslash [j,(j+s-1) \text{ mod } N]$, be some message demanded by the user $S_j$. We prove that, for each $i_1 \in [0,z-1]$, the user $S_j$ can retrieve the $i_1^{th}$ block of the demanded message $x_l$, i.e., $x_l^{i_1}$, using the transmissions done during the iteration $(l-i_1(s-1)) \text{ mod } N$.

Recall that during each iteration $i \in [0,N-1]$, 
each user $S_{(k(s-1)+ i) \text{ mod } N}$ transmits one coded symbol $Y_k^i$ where the $k^{th} $ block of the first message $x_{(k(s-1)+ i) \text{ mod } N}$ and the $(k+1)^{th}$ block of the last message $x_{((k+1)(s-1)+ i) \text{ mod } N}$ available as side information are taken to be included in the coded symbol. 
Hence, the $i_1^{th}$ block of the demanded message $x_l$, i.e., $x_l^{i_1}$, is present in the coded symbol $Y_{i_1}^{(l-i_1(s-1)) \text{ mod } N}$ since,
\begin{align}
Y_{i_1}^{(l-i_1(s-1)) \text{ mod } N}=&x_{(i_1(s-1)+(l-i_1(s-1))) \text{ mod } N}^{i_1} \oplus\\& x_{((i_1+1)(s-1)+(l-i_1(s-1))) \text{ mod } N}^{i_1+1}\\
=&x_{l}^{i_1} \oplus x_{(l+s-1) \text{ mod } N}^{i_1+1}.
\end{align}
Now, we need to find a message $x_{l'}$ in the side-information set of  the user $S_{j}$ such that  $x_{l'}$ is taken at an interval of some multiple of $(s-1)$  starting from $x_l$, i.e., $$(l'-l) \text{ mod } N =t(s-1), \text{ for some }t \in \mathbb{Z}.$$
We observe that the smallest $t$ possible such that $x_{l'}$ is in the side-information set of  the user $S_{j}$,  is $t_1=\left \lceil \frac{(j-l) \text{ mod } N}{s-1} \right \rceil$.
%\begin{align*}
%j-l&=(s-1)j_1+j_2\\
%l'_1 &=
%\begin{cases*}
%j_1, & \text{ if $j_2=0$} \\
%j_1+1, & \text{ otherwise.}
%\end{cases*}
%\end{align*}

We prove that for each  $i_1 \in [0,z-1-t_1]$, the user $S_j$ can retrieve the block $x_l^{i_1}$ from $D_l^{i_1}, $ where, $l_1= (l-i_1(s-1)) \text{ mod } N$ and
\begin{equation*}
D_l^{i_1} =
\bigoplus_{k \in [i_1, i_1+t_1-1]}  Y_k^{l_1}.
\end{equation*}
This is since, $x_{(l+ t_1 (s-1) ) \text{ mod } N}$ is available at the user $S_j $ and 
	\begin{align*}
	D_l^{i_1} &= \bigoplus_{k \in [i_1 , i_1+t_1-1]}  Y_k^{l_1} \\
	&=  \bigoplus_{k \in [i_1 , i_1+t_1-1]}  x_{(k(s-1)+l_1) \text{ mod } N}^k \oplus\\&\hspace{1cm} x_{((k +1) (s-1) +l_1) \text{ mod } N}^{k+1} \\
	&=  x_{(i_1(s-1)+l_1) \text{ mod } N}^{i_1} \oplus x_{((i_1+t_1) (s-1) +l_1) \text{ mod } N}^{i_1+t_1} \\
	&=  x_{l}^{i_1} \oplus \underbrace{x_{(l+t_1 (s-1)) \text{ mod } N}^{i_1+t_1}. }_{\text{available as side-information}}
	\end{align*}
	Hence, the user $S_j$ can retrieve $x_{l}^{i_1}$ for each $i_1 \in [0,z-1-t_1].$
	 
	 Now, consider the case when $i_1 \in [z-t_1, z-1]$. 
	 Here also, we find a message $x_{l'}$ in the side-information set of  the user $S_{j}$ such that  $x_{l}$ is taken at an interval of some multiple of $(s-1)$  starting from $x_{l'}$, i.e., $$(l-l') \text{ mod } N =t'(s-1), \text{ for some }t' \in \mathbb{Z}.$$
	 We observe that the smallest $t'$ possible such that $x_{l'}$ is in the side-information set of  the user $S_{j}$,  is $z-t_1=z-\left \lceil \frac{(j-l) \text{ mod } N}{s-1} \right \rceil$. Also, recall that the $i_1^{th}$ block of the demanded message $x_l$, i.e., $x_l^{i_1}$, is present in the coded symbol $Y_{i_1-1}^{(l-i_1(s-1)) \text{ mod } N}$ since,
	 \begin{align}
	 Y_{i_1-1}^{(l-i_1(s-1)) \text{ mod } N}=&x_{((i_1-1)(s-1)+(l-i_1(s-1))) \text{ mod } N}^{i_1-1} \oplus\\& x_{(i_1(s-1)+(l-i_1(s-1))) \text{ mod } N}^{i_1}\\
	 =&x_{(l-(s-1)) \text{ mod } N}^{i_1-1} \oplus x_{l}^{i_1}.
	 \end{align}
	 
	 So, we prove that for each  $i_1 \in [z-t_1, z]$, the user $S_j$ can retrieve the block $x_l^{i_1}$ from $D_l^{i_1}, $ where $l_1= (l-i_1(s-1)) \text{ mod } N$ and
	 \begin{equation*}
	 D_l^{i_1} =
	 \bigoplus_{k \in [i_1-(z-t_1), i_1-1]}  Y_k^{l_1}.
	 \end{equation*}
%	\item If $i \in [z-t,z]$, then
	This is since, $x_{(l- (z-t_1) (s-1) ) \text{ mod } N}$ is available at the user $S_j $ and 
	\begin{align*}
	D_l^{i_1} &= \bigoplus_{k \in [i_1-(z-t_1), i_1-1]}  Y_k^{l_1} \\
	&= \bigoplus_{k \in [i_1-(z-t_1), i_1-1]} x_{(k(s-1)+l_1) \text{ mod } N}^k \oplus\\&\hspace{1cm} x_{((k +1) (s-1) +l_1) \text{ mod } N}^{k+1} \\
	&=  x_{((i_1-(z-t_1))(s-1)+l_1) \text{ mod } N}^{i_1-(z-t_1)} \oplus x_{(i_1 (s-1) +l_1) \text{ mod } N}^{i_1} \\
	&= \underbrace{x_{(l-(z-t_1)(s-1)) \text{ mod } N}^{i_1-(z-t_1)}}_{\text{available as side-information}} \oplus  x_{l}^{i_1}
	\end{align*}
Hence, the user $S_j$  can retrieve $x_{l}^{i_1}$ for each $i \in [z-t_1,z].$
%\end{itemize}
Therefore, we proved that the user $S_j $ can retrieve the message $x_{l}$ in its demand set. Similarly, all other messages in $W_j$ can be retrieved.

\textit{Proof of Theorem \ref{thm: main result}:}
The total number of bits transmitted is $\frac{Nd}{\left \lceil \frac{s}{N-s} \right \rceil}$ bits for Case A.  Hence, the normalized rate is $\frac{1}{\left \lceil \frac{s}{N-s} \right \rceil}$.  The total number of bits transmitted is $\frac{\left \lceil \frac{N-s}{s-1} \right \rceil Nd}{1+\left \lceil \frac{N-s}{s-1} \right \rceil }$ bits for Case B. Hence, the normalized rate is $\frac{\left \lceil \frac{N-s}{s-1} \right \rceil }{1+\left \lceil \frac{N-s}{s-1} \right \rceil }$.  This completes the proof. 

\section{conclusion} \label{conclusion}
In this paper, we have explored a specific classes of EICP, namely, consecutive and symmetric EICP. We have provided code construction for this case. By efficiently  utilizing the  sub-packetization scheme, we were able to achieve a normalized rate lower than that of the state of the art \cite{potter2019eic,anjana2020eic} for some cases. For other cases, we conjecture that the normalized rate achieved using our scheme is lower than that of the state of the art \cite{potter2019eic,anjana2020eic}. 
%We have also studied the task based extension of CS-EICP.
% We have provided code construction for task based CS-EICP, when the cardinality of the demand set is one. 
 In this paper, we had only explored a specific class of EICP. Explicit code construction for general EICP is still open. Exploring techniques to find a general solution is an interesting thing to work on. 
%%%%%%%%%%%%%%%%%%%%%%%%%%
\section*{Acknowledgement}
This work was supported by the Science and Engineering Research Board (SERB) of Department of Science and Technology (DST), Government of India, through J. C. Bose National Fellowship to B. Sundar Rajan. This work was done when Shanuja Sasi was at Purdue University as a visiting scholar under Science and Engineering Research Board (SERB) Overseas Visiting Doctoral Fellowship (OVDF).

\bibliographystyle{IEEEtran}
%\bibliography{ref}

\end{document}